\documentclass[11pt,reqno]{amsart}
\allowdisplaybreaks[4]
\RequirePackage[OT1]{fontenc}
\RequirePackage{amsthm,amsmath}
\RequirePackage[numbers]{natbib}
\RequirePackage[colorlinks,linkcolor=blue,citecolor=blue,urlcolor=blue]{hyperref}
\usepackage{amssymb}
\usepackage{anysize}
\usepackage{hyperref}
\usepackage{extarrows}
\usepackage{multirow}
\usepackage{natbib}
\usepackage{stmaryrd}
\usepackage{color}
\usepackage{cancel}
\usepackage{arydshln}
\usepackage{bbm}

\numberwithin{equation}{section}
\newtheorem{thm}{Theorem}[section]
\newtheorem{lem}[thm]{Lemma}
\newtheorem{pro}[thm]{Proposition}
\newtheorem{cor}[thm]{Corollary}

\newtheorem{rem}[thm]{Remark}
\newtheorem{assu}[thm]{Assumption}

\newcommand{\be}{\begin{equation}}
\newcommand{\ee}{\end{equation}}
\newcommand{\bea}{\begin{eqnarray*}}
\newcommand{\eea}{\end{eqnarray*}}

\newcommand{\lone}{\mathbbm{1}}

\newtheorem{defn}[thm]{Definition}
\theoremstyle{remark}

\makeatletter

\newcommand{\Rmnum}[1]{\expandafter\@slowromancap\romannumeral #1@}
\makeatother
\makeatletter % `@' now normal "letter"
\@addtoreset{equation}{section}
\makeatother  % `@' is restored as "non-letter"

\renewcommand{\thefootnote}{\fnsymbol{footnote}}

\newcommand{\beqa}{\begin{eqnarray}}
\newcommand{\eeqa}{\end{eqnarray}}

\renewcommand{\Im}{\mathsf{Im}}

\newcommand{\Tr}{\mathrm{Tr}}
\newcommand{\caH}{\mathcal{H}}

\newcommand{\ii}{\mathrm{i}}
\newcommand{\dd}{\mathrm{d}}

\newcommand{\im}{{\mathfrak{Im} \, }}
\newcommand{\E}{{\mathbb E }}
\newcommand{\R}{{\mathbb R }}
\newcommand{\N}{{\mathbb N}}

\newcommand{\Z}{{\mathbb Z}}
\renewcommand{\P}{{\mathbb P}}

\newcommand{\C}{{\mathbb C}}

\renewcommand{\im}{\Im}

\newcommand{\nc}{\normalcolor}

\begin{document}

 \begin{minipage}{0.85\textwidth}
 \vspace{2.5cm}
 \end{minipage}
\begin{center}
\large\bf
 Equipartition principle for Wigner matrices 
\end{center}

\renewcommand{\thefootnote}{\fnsymbol{footnote}}	
\vspace{1cm}
\begin{center}
 \begin{minipage}{0.32\textwidth}
\begin{center}
Zhigang Bao\footnotemark[1]  \\
\footnotesize {HKUST}\\
{\it mazgbao@ust.hk}
\end{center}
\end{minipage}
\begin{minipage}{0.32\textwidth}
\begin{center}
L\'aszl\'o Erd{\H o}s\footnotemark[2]  \\
\footnotesize {IST Austria}\\
{\it lerdos@ist.ac.at}
\end{center}
\end{minipage}
\begin{minipage}{0.33\textwidth}
 \begin{center}
Kevin Schnelli\footnotemark[3]\\
\footnotesize 
{KTH Royal Institute of Technology}\\
{\it schnelli@kth.se}
\end{center}
\end{minipage}
\footnotetext[1]{Supported in parts by Hong Kong RGC Grant GRF 16301519, and NSFC 11871425 }
\footnotetext[2]{Supported in parts by ERC Advanced Grant RANMAT No.\ 338804.}
\footnotetext[3]{Supported in parts by the Swedish Research Council Grant VR-2017-05195 and the Wallenberg AI, Autonomous Systems and Software Program.}

\renewcommand{\thefootnote}{\fnsymbol{footnote}}	

\end{center}

\vspace{1cm}

\begin{center}
 \begin{minipage}{0.83\textwidth}\footnotesize{
 {\bf Abstract.}   
  We prove that the energy of any eigenvector of a sum of several independent large Wigner matrices
   is equally distributed among these  matrices
 with very high precision. This shows a particularly strong microcanonical form of
 the equipartition principle for quantum systems whose components are modelled by Wigner matrices.
}
\end{minipage}
\end{center}

 \vspace{5mm}
 
 {\small
\footnotesize{\noindent\textit{Date}: August 17, 2020}
 
 \vspace{2mm}

 }

\thispagestyle{headings}

\section{Introduction}

 Equipartition of energy is a general principle in classical statistical physics stating that in 
an ergodic system at equilibrium  the total energy is shared equally among the elementary
degrees of freedom. In quantum systems equipartition
 breaks down  at very low temperatures. Even at 
higher temperatures there is no general quantum counterpart of this principle apart from the standard quantum virial theorem that only relates the total kinetic energy 
to a certain derivative of the potential.
Nevertheless, in some special cases this principle
 could be verified; see \cite{BSL19} and references therein for
an extensive physics literature on the popular model of a single quantum particle
in contact with a quantum heat bath consisting of infinitely many harmonic oscillators.  
In the current paper we show that for Wigner random matrices, i.e. for a mean-field quantum system
with random quantum transition rates,  a particularly strong microcanonical 
form of the quantum equipartition holds: it is valid separately for every eigenvector.  

More precisely,  suppose that  the total Hamiltonian of a quantum system 
is represented by a sum of independent $N\times N$ Wigner matrices $ H = H_1 + H_2+ \ldots + H_k$, where
each $H_\iota$ represents the Hamiltonian of a subsystem.
Let $w=(w(1), \ldots, w(N))^\top\in \C^N$ be an $\ell^2$-normalized eigenvector of $H$ with eigenvalue $\lambda$, i.e.\ $Hw=\lambda w$.
The eigenvalue $\lambda$ is the total energy of $w$:
\begin{align*}
    \lambda= E(w):= (w, H w) = \sum_{\iota=1}^k ( w, H_\iota w)\,.
    \end{align*}
The energy of the $\iota$-th subsystem $H_\iota$ in the state given by $w$ is $E_\iota(w):=( w, H_\iota w)$. 
Our main result asserts that 
\begin{equation}\label{eqp}
   E_\iota(w) \approx \frac{E(w)}{k}, \qquad \forall \iota=1, 2,\ldots k\,,
\end{equation}
with very high precision and with very high probability. In other words,
 the total energy is equally distributed among the $k$ subsystems.

Fine properties of eigenvectors of large Wigner matrices have been extensively studied in the recent years.
They are {\it delocalized} i.e.\ $\max_i |w(i)|\le N^{-1/2+\epsilon}$ for any fixed $\epsilon>0$ with 
very high probability as $N$ tends to infinity. Delocalization  follows
directly from the optimal local law, see e.g.\ \cite{EKYY}, and~\cite{BeLo20} for an optimal rate.
 Moreover, the eigenvectors are asymptotically normal,
in the sense that for any fixed deterministic vector $q\in \C^N$ the moments of $\sqrt{N}|( q, w)|$ 
coincide with those of the modulus of a standard Gaussian \cite{BY, KY, TV}.
A multi-variate extension involving the joint moments
of several eigenvectors also holds \cite{BY}. Furthermore, the {\it quantum unique ergodicity}  is also valid, 
stating that 
\begin{equation}\label{que}
\sum_{i\in J} |w(i)|^2 \approx \frac{|J|}{N}\,,
\end{equation}
 for any deterministic subset $J\subset\{ 1,2, \ldots, N\}$; see 
\cite{Be19,BY,MY20}.  The  key difficulty in these latter results was to prove them 
{\it microcanonically}, i.e.\ for each eigenvector; this required the 
sophisticated equilibration mechanism of the Dyson Brownian motion. In contrast, the local law 
(see Theorem~\ref{theorem local law1} later) directly implies
the analogous results 
for a spectral projection on mesoscopic scale, e.g.
\begin{equation}\label{queav}
\frac{1}{2N^\epsilon}
  \sum_{|\alpha-\alpha_0|\le N^\epsilon}  \sum_{i\in J} |w_\alpha(i)|^2 \approx \frac{|J|}{N}\,,
\end{equation}
instead of \eqref{que}, 
 involving an average over many eigenvectors 
$w_\alpha$ with eigenvalues $\lambda_\alpha$ near $\lambda_{\alpha_0}$ with a fixed $\alpha_0$.
Here the eigenvalues $\lambda_\alpha$ are indexed in an increasing order, $\lambda_1\le \lambda_2\le \ldots \le\lambda_N$.

In all these previous results the eigenvector was tested against a specific {\it deterministic} observable; while in
the equipartition  relation \eqref{eqp} we consider the quadratic form of $w$ with a {\it random} $H_\iota$
that is far from being independent of $w$. 
Given the complicated dependence between $w$ and $H_\iota$ 
it is somewhat surprising that the proof of \eqref{eqp} is simpler 
 than that of \eqref{que}.
In fact, despite this dependence, 
we can still directly handle $(w, H_\iota w)$ for an individual eigenvector, i.e.  
we {\it do not need} to establish first a spectrally local-averaged version of \eqref{eqp} in the form
\begin{align*}
 \frac{1}{2N^\epsilon}
  \sum_{|\alpha-\alpha_0|\le N^\epsilon} (w_\alpha, H_\iota w_\alpha)\approx \frac{\lambda_{\alpha_0}}{k}
\end{align*}
 and then prove that
$ ( w_\alpha, H_\iota w_\alpha)$ does not change much if the eigenvalue $\lambda_\alpha$
remains close to a fixed energy. 

The main reason for the simple proof is algebraic. Consider $k=2$ for simplicity. It turns out that the quadratic
forms of $\mathcal{H}:=H_1- H_2$ are especially small due to a strong algebraic cancellation in the cumulant expansion.
Once the smallness of $( w, \mathcal{H} w) = ( w, H_1w) - ( w, H_2w)$
is established, \eqref{eqp} follows from $\lambda =  ( w, H_1w) + (w, H_2w)$.

To demonstrate the central role of $\mathcal{H}$, in the next section
we first give the proof of \eqref{eqp} for $k=2$ in the Gaussian case,
where the mechanism is especially elementary. Then we introduce the general model and properly state our
result in Section~\ref{sec:def}. After collecting  some preliminaries from earlier papers in Section~\ref{sec:prelim},
we will prove our main theorem starting in Section~\ref{section proof of maintheorem} for the complex Hermitian  case
under the additional condition $\E h_{\iota,ij}^2=0$  on  the entries of each matrix~$H_\iota$. 
This condition is removed in Section~\ref{remark about variance complex case}.
The necessary modifications for the real symmetric case are presented in Section~\ref{le real symmetric case}.

\section{A simple proof of \eqref{eqp} for the Gaussian case and $k=2$}\label{sec:gauss}

Assume we are given two independent GUE random matrices $H_{1}$ and $H_{2}$ of size $N\times N$,
i.e.\ their  entries are two sets of independent complex centered Gaussian random variables of variance $\frac{1}{2N}$ subject to the symmetry constraint $H_1=H_1^*$ and $H_2=H_2^*$. Then clearly the sum
\begin{align}
 H:=H_1+H_2
\end{align}
also belongs to the standard Gaussian unitary ensemble (GUE).
Denote by $(\lambda_\alpha)_{\alpha}$ the eigenvalues in ascending order of $H$ and let $(w_\alpha)_{\alpha}$ be an associated normalized eigenbasis, i.e.\  we have $(w_\alpha, Hw_\beta)=\delta_{\alpha\beta}\lambda_\alpha$, for any choice of indices $\alpha,\beta$. 

Consider now the random variables
\begin{align*}
( w_\alpha, H_1w_\beta)-\frac{\delta_{\alpha\beta}}{2}\lambda_\alpha\,.
\end{align*}
We claim that, for any $N$, these random variables are Gaussian.
\begin{lem}\label{le lemma gaussian}
For any choice of index $\alpha$ the random variable 
\begin{align}\label{le lemma gaussian equation}
 (w_\alpha,H_1w_\alpha)-\frac{1}{2}\lambda_\alpha
 \end{align}
 is a centered real Gaussian random variable with variance $\frac{1}{4N}$, for any $N$. Moreover, for any choice of indices $\alpha$ and $\beta$, with $\alpha\not=\beta$, the random variable
 \begin{align}\label{le lemma gaussian equation 2}
  (w_\alpha,H_1w_\beta)
 \end{align}
is a centered complex Gaussian random variable of variance $\frac{1}{4N}$, for any $N$.
\end{lem}

\begin{proof}
Introduce the auxiliary matrix
\begin{align}\label{le easy auxH}
\mathcal{H}:=H_1- H_2\,,
\end{align}
whose entries are also independent centered Gaussian random variables, up to the symmetry constraint, with variance $\E|\mathcal{H}_{ij}|^2=\frac{1}{N}$. 
A simple calculation then shows that $\E [H_{ij}\mathcal{H}_{ab}]=0$, for all $i,j,a,b\in\llbracket 1,N\rrbracket$, hence the matrices $H$ and $\caH$ are independent. In particular, $\caH$ is independent from $w_\alpha$ and $w_\beta$, for any choice of $\alpha,\beta$.

Observe now that we can write the random variables in~\eqref{le lemma gaussian equation} and~\eqref{le lemma gaussian equation 2} as
\begin{align}\label{le baby computation}
(  w_\alpha, H_1 w_\beta)-\frac{1}{2}\lambda_\alpha\delta_{\alpha\beta}&=( w_\alpha , H_1 w_\beta)-\frac12 (w_\alpha, H_1w_\beta)-\frac12(w_\alpha,H_2w_\beta)=\frac{1}{2}(w_\alpha,\mathcal{H}w_\beta)\,.
\end{align}
Hence by the independence of $H$ and $\caH$ we conclude that $\frac{1}{2}w_\alpha^*\mathcal{H}w_\beta$ is a Gaussian random variable. Since $\E \caH_{ij}=0$, it follows that $\E w_\alpha^*\mathcal{H}w_\beta=0$. Further we have
\begin{align*}
 \E |w_\alpha^*\mathcal{H}w_\beta|^2=\sum_{ijab}\E\overline{w_\alpha(i)}\caH_{ij}w_\beta(j)w_\alpha(a)\overline{\caH_{ab}}\,\overline{w_\beta(b)}=\sum_{ijab}\frac{1}{N}\delta_{ia}\delta_{jb}|w_\alpha(a)|^2|w_\beta(j)|^2=\frac{1}{N}\,,
\end{align*}
where we used independence and that the eigenvectors are $\ell^2$-normalized. 
 The notation $\sum_{ijab}$ means that we sum over all indices from $1$ to $N$.
This shows~\eqref{le lemma gaussian equation} and~\eqref{le lemma gaussian equation 2}.
\end{proof}

{\it Notation:} 
The symbol $O(\,\cdot\,)$ stands for the standard big-O notation. We use~$c$ and~$C$ to denote positive finite constants that do not depend on the matrix size~$N$. Their values may change from line to line. 
We use double brackets to denote index sets, i.e.\ for $n_1, n_2\in\R$, $\llbracket n_1,n_2\rrbracket:= [n_1, n_2] \cap\Z$.

For vectors $v,w\in \mathbb{C}^N$, we write $v^*w=(v,w)$ for their scalar product. For an $N$ by $N$ matrix~$A$, we denote by $\|A\|$ its operator norm and by $\|A\|_\infty:=\max_{ij}|A_{ij}|$.  We use $\langle A\rangle :=\frac{1}{N}\sum_{i} A_{ii}$ to denote the normalized trace of an $N\times N$ matrix $A=(A_{ij})_{N,N}$.

\section{Definitions and results}\label{sec:def}

In this section we introduce the model and our main result on equipartition.
\begin{assu}\label{assumption 1}
 Fix an integer $k\ge 2$. Let  $H_\iota:=(h_{\iota,ij})$, $\iota=1,2,\ldots k$, be $k$ independent complex Hermitian Wigner matrices of size $N\times N$, i.e.,  we assume that their entries are independent centred random variables, up to the symmetry constraints $h_{\iota,ij}=\overline{h_{\iota,ji}}$, satisfying
\begin{align}\label{le variance complex case}
\mathbb{E}|h_{\iota,ij}|^2&=\frac{1}{N}\,,\qquad\qquad 1\leq i,j\leq N\,,\qquad \iota=1,\ldots k\,,
\end{align}
and the families of random variables $\{h_{\iota,ij}\}$ have finite moments to all order, i.e., for each $m\ge 3$ there is a positive constant $C_m$ such that
\begin{align}\label{entries moment bounds}
 \E|\sqrt{N}h_{\iota,ij}|^m\le C_m\,,\qquad\qquad m\ge 3\,,\qquad \iota=1,2,\ldots k\,.
\end{align}
\end{assu}

For the main part of the paper we 
 assume that $H_\iota$  are complex Hermitian matrices.  This assumption  is only for simplicity
 of the presentation; our result holds and the proof  also applies \nc
  with minor changes to the real symmetric setup as well; see Remark~\ref{remark 1}.

Choose now $k$  possibly $N$-dependent numbers $\sigma_\iota\ge 0$ such that 
\begin{align}\label{le sigma condition}
\sum_{\iota=1}^k\sigma_\iota^2=1\,,
\end{align}
and consider the random matrix
\begin{align}\label{le H}
H:= \sum_{\iota=1}^k\sigma_\iota H_\iota\,.
\end{align}
To present our results, we use the following definition of high-probability estimates.\nc

\begin{defn}\label{definition of stochastic domination}
Let $\mathcal{X}\equiv \mathcal{X}^{(N)}$ and $\mathcal{Y}\equiv \mathcal{Y}^{(N)}$ be two sequences of
 nonnegative random variables. We say that~$\mathcal{Y}$ stochastically dominates~$\mathcal{X}$ if, for all (small) $\epsilon>0$ and (large)~$D>0$,
\begin{align}\label{le proba est}
\P\big(\mathcal{X}^{(N)}>N^{\epsilon} \mathcal{Y}^{(N)}\big)\le N^{-D},
\end{align}
for sufficiently large $N\ge N_0(\epsilon,D)$, and we write $\mathcal{X} \prec \mathcal{Y}$ or $\mathcal{X}=O_\prec(\mathcal{Y})$.
 When
$\mathcal{X}^{(N)}$ and $\mathcal{Y}^{(N)}$ depend on a parameter $v\in \mathcal{V}$ (typically an index label or a spectral parameter), then $\mathcal{X}(v) \prec \mathcal{Y} (v)$, uniformly in $v\in \mathcal{V}$, means that the threshold $N_0(\epsilon,D)$ can be chosen independently of $v$. 
\end{defn}
We often use the notation $\prec$ also for deterministic quantities, then $\mathcal{X}^{(N)}\le N^{\epsilon} \mathcal{Y}^{(N)}$ holds with probability one. Stochastic domination has the following properties.
\begin{lem}\label{dominant}(Proposition 6.5 in \cite{book})
	\begin{enumerate}
		\item $X \prec Y$ and $Y \prec Z$ imply $X \prec Z$;
		\item If $X_1 \prec Y_1$ and $X_2 \prec Y_2$, then $X_1+X_2 \prec Y_1+Y_2$ and $X_1X_2 \prec Y_1Y_2;$
		\item  If $X \prec Y$, $\E Y \geq N^{-c_1}$ and $|X| \leq N^{c_2}$ almost surely with fixed constants $c_1$ and $c_2$, then we have $\E X \prec \E Y$.
	\end{enumerate}
\end{lem}

Let $(\lambda_\alpha)_\alpha$ be the eigenvalues of the matrix $H$ in ascending order
and let $(w_\alpha)_\alpha$ be a basis of associated normalized eigenvectors. In this paper we are interested in estimating 
\begin{align}
 w_\alpha^* H_\iota w_\beta-  \sigma_\iota\lambda_\alpha\delta_{\alpha\beta}\,,\qquad \iota=1,\ldots, k\,,
\end{align}
for any choice of $\alpha,\beta\in\llbracket 1,N\rrbracket$.

\begin{thm}\label{le main theorem}
Let $H$ be given by~\eqref{le H}, and assume $H_\iota$, $\iota=1,\ldots, k$, satisfy Assumption~\ref{assumption 1} and that $\sigma_\iota$, $\iota=1,\ldots k$, satisfy~\eqref{le sigma condition}. Then  
\begin{align}\label{le main theorem equation}
 \Big|w_\alpha^*H_\iota w_\beta-  \sigma_\iota\lambda_\alpha\delta_{\alpha\beta}\Big|\prec\frac{1}{\sqrt{N}}\,,
\end{align}
for all $\alpha,\beta\in\llbracket 1,N\rrbracket$  and $\iota\in \llbracket 1,k\rrbracket$.
\end{thm}
\nc

\begin{rem}\label{remark 1}
We formulated Theorem~\ref{le main theorem} for complex Hermitian Wigner matrices, but with some modifications our method and results carry over the real symmetric case; see Theorem~\ref{le theorem for symmetric} below. The details are given in Section~\ref{le real symmetric case}.

We further remark that one may also consider a mixed symmetry setup where some $H_\iota$'s are complex Hermitian Wigner matrices while the remaining $H_\iota$ are real symmetric Wigner matrices. The arguments in Section~\ref{le real symmetric case} can be extended to such a setting and~\eqref{le main theorem equation} indeed holds under this setup as well.
\normalcolor
\end{rem}

\section{Preliminaries}\label{sec:prelim}
In this section we collect some essential tools used in the proof of Theorem~\ref{le main theorem}. We start with the Green function of the random matrix $H$ and the corresponding local laws.

\subsection{Local law for the Green function and rigidity of eigenvalues}

For any probability measure $\mu$ on $\R$, its Stieltjes transform is defined as 
\begin{align}\label{stieltjes transform}
m_\mu(z):=\int_{\mathbb{R}} \frac{1}{x-z}\, \dd\mu(x)\,, \qquad \qquad z\in \mathbb{C}\backslash\mathbb{R}\,.
\end{align}
We denote the Stieltjes transform of the standard semicircular law by $m_{sc}(z)$.

 Let $G$ denote the Green function or resolvent of $H$, i.e.\
\begin{align}\label{green function}
 G(z):= \frac{1}{H-z}\,,\qquad z\in\C\backslash\R\,.
\end{align}
We refer to $z=E+\ii\eta$ in~\eqref{stieltjes transform} and~\eqref{green function} as spectral parameter. We denote by $m(z)$ the normalized trace of $G(z)$, i.e.,
\begin{align}\label{le m}
 m(z)=\frac{1}{N}\Tr G(z)=\langle G(z)\rangle \,,\qquad z\in\mathbb{C}\backslash\mathbb{R}\,,
\end{align}
and note that by the spectral calculus $m(z)$ is the Stieltjes transform of the empirical eigenvalue distribution of $H$. Finally, we recall the deterministic estimate $\|G(z)\|_\infty\le \|G(z)\|\le |\eta|^{-1}$ with $\eta=\im z$.

We are interested for energies $E$ in a neighborhood of the support of the semicircular law, i.e. $|E|<2+\varrho$, for some fixed $\varrho>0$. Further, fix a small $\epsilon>0$, and introduce the spectral domain
\begin{align}\label{le domain}
\mathcal{E}:=\{z=E+\ii\eta\in\C\,:\, E\in[-2-\varrho,2+\varrho],  N^{-1+\epsilon}\le|\eta|\le 1\}\,.
\end{align}

For $z,z'\in\mathcal{E}$, let $\Psi(z,z')$ denote the deterministic control parameter
\begin{align}\label{le Psi}
 \Psi(z,z'):=\frac{1}{\sqrt{N\eta_0}}\,,\qquad \eta_0=\min\{|\im z|,|\im z'|\}\,,
\end{align}
and we use the convention $\Psi(z,z)\equiv\Psi(z)$.

Let $\gamma_\alpha$ be the $\alpha$-th $ N$-quantile of the semicircle law, i.e.\ $\gamma_\alpha$ is determined by
\begin{align}\label{quantile}
 \int_{-2}^{\gamma_\alpha} \frac{1}{2\pi}\sqrt{4-x^2}\,\dd x= \frac{\alpha-1/2}{N}\,.
\end{align}
The quantile $\gamma_\alpha$ is often also referred to as the {\it classical location} of the eigenvalue $\lambda_\alpha$. 

One ingredient for our work is the following strong local law for the Green function and the eigenvalue rigidity estimate.

\begin{thm}[Theorem 2.1 in~\cite{EYY}, Theorem 2.3 in~\cite{EKYY}]\label{theorem local law1}
 Let $H$ be as in~\eqref{le H} satisfying Assumption~\ref{assumption 1}. Then we have the uniform estimates
 \begin{align}\label{le local law}
  |G_{ij}(z)-\delta_{ij}m_{sc}(z)|\prec \left({\frac{\im\, m_{sc}(z)}{N\eta}}\right)^{1/2}+\frac{1}{N|\eta|} \prec \Psi(z)\,,\qquad |m(z)-m_{sc}(z)|\prec \Psi(z)^2\,,
 \end{align}
  for all $z=E+\ii\eta\in\mathcal{E}$. 
  
Moreover we have the eigenvalue rigidity estimate
\begin{align}\label{le rigidity estimate}
|\lambda_\alpha-\gamma_\alpha|\prec\frac{1}{N^{2/3}\min\{\alpha,N-\alpha+1\}^{1/3}}\,, 
\end{align}
for all $\alpha\in\llbracket1,N\rrbracket$.
\end{thm}

\subsection{Cumulant expansion}

A second main tool in the proof of Theorem~\ref{le main theorem} are cumulant expansions which were for example used in~\cite{KKP,LP} to study linear eigenvalue statistics of random matrices. For our purposes the following version from~\cite{heko,hero} is very suitable.

\begin{lem}[Lemma 2.4 in~\cite{hero}, Lemma 7.1 in~\cite{heko}]\label{complex_cumulant}
	Let $h$ be a complex-valued random variable with finite moments. Let $\kappa^{(p,q)}$ be the $(p,q)$ cumulant of $h$, which is defined as 
	\begin{align}\label{le complex cumulant basic}
	\kappa^{(p,q)}:=(-\ii)^{p+q} \Big( \frac{\partial^{p+q}}{\partial s^p \partial t^q} \log \E\, \mathrm{e}^{\ii s h+\ii t \overline{h}} \Big) \Bigg|_{s,t=0}.
	\end{align}
	Let $f\in C^{\infty}(\C^2;\C)$, then for any fixed $l \in \N$, we have
	
	\begin{align}\label{le simple cumulant}\E h f(h, \overline{h})=\sum_{p+q=0}^l \frac{1}{p!q!} \kappa^{(p+1,q)} f^{(p,q)}(h,\overline{h}) +\Omega_{l+1}\,,\end{align}
	where
	\begin{align*}
	 f^{(p,q)}(w_1,w_2):=\partial_{w_1}^p\partial_{w_2}^q f(w_1,w_2)\,,\qquad w_1,w_2\in\C\,,
	\end{align*}
and the error term $\Omega_{l+1}$ satisfies
	\begin{align}\label{le error in complex cumulant}
	|\Omega_{l+1}| &\le C_l \E \Big[|h|^{l+2}\Big] \max_{p+q=l+1}\sup_{|w| \leq M} |f^{(p,q)}(w,\overline{w})| \nonumber\\
	&\qquad\qquad+C_l\Big[ \E \big[ |h|^{2l+4} \lone_{ |h|>M }
	\big]  \E\big[\max_{p+q=l+1}\sup_{|w|\le |h|}|f^{(p,q)}(w,\overline{w})|^2\big]\Big]^{1/2}\,,
	\end{align}
	where $M>0$ is an arbitrary cutoff.
\end{lem}
We remark Lemma~\ref{complex_cumulant} is a combination of Lemma 2.4 in~\cite{hero} and Lemma 7.1~\cite{heko}; the combinatoric part comes from~\cite{heko} and
the error estimate is taken from~\cite{hero}.

From~\eqref{le complex cumulant basic}, the first few complex cumulants are given by
\begin{align*}
 \kappa^{(1,0)}=\E h\,,\qquad \kappa^{(1,1)}=\E |h|^2-|\E h|^2\,,\qquad \kappa^{(2,0)}=\E h^2-(\E h)^2\,,  
\end{align*}
etc., with $\kappa^{(q,p)}=\overline{\kappa^{(p,q)}}$.

\section{Proof of Theorem~\ref{le main theorem}}\label{section proof of maintheorem}

  The proof of Theorem~\ref{le main theorem} is based on an essentially optimal estimate on a distinguished observable we introduce in this section; see~\eqref{le observable} below. We are going to prove Theorem~\ref{le main theorem} for the case $k=2$, the case of general $k$ then follows easily by grouping all but one summands in
   \eqref{le H} together and viewing it as a single Wigner matrix.

Generalizing~\eqref{le easy auxH}, we introduce the auxiliary matrix
\begin{align}\label{le auxH}
\mathcal{H}:=\sigma_2 H_1-\sigma_1 H_2\,,
\end{align}
whose entries are independent centered random variables, up to the symmetry constraint, with variance $\E|\mathcal{H}_{ij}|^2=\frac{1}{N}(\sigma_{1}^2+\sigma_{2}^2)=\frac1N$; see~\eqref{le sigma condition}.
 In order to prove Theorem~\ref{le main theorem}, we derive a high moment estimate for observables of the form
\begin{align}\label{le observable}
\frac{1}{N}\text{Tr} \mathcal{H} \Im G(z_1)\mathcal{H} \Im G(z_2)=\langle\caH\Im G(z_1)\caH \Im G(z_2)\rangle\,,\qquad\qquad z_1,z_2\in\mathcal{E}\,,
\end{align}
where $G$ denotes the Green function of $H$; see~\eqref{green function} and the set $\mathcal{E}$ was
defined in~\eqref{le domain}. The main technical result of this paper is the following proposition.
\begin{pro}\label{le main proposition} Under the assumptions of Theorem~\ref{le main theorem}, we have the estimate
\begin{align}\label{le claim}
 \langle\mathcal{H} \Im G(z_1)\mathcal{H} \Im G(z_2)\rangle\prec 1\,,
\end{align}
uniformly in $z_1,z_2\in\mathcal{E}$. 
\end{pro}

\begin{rem}
 Using the deterministic bound $\|G(z)\|\le\frac{1}{|\eta|}$ and the bounds $\|H_1\|$, $\| H_2\|\prec 1$, which follow from~\eqref{le rigidity estimate}, we get the a priori bound
 \begin{align}
\langle \mathcal{H} \Im G(z_1)\mathcal{H} \Im G(z_2)\rangle\prec\frac{1}{|\im z_1|\,|\im z_2|}\prec N^2\,,
 \end{align}
 on the spectral domain $\mathcal{E}$. Thus~\eqref{le claim} is an improvement of two orders in $N$ and gives the correct size, up to factors of $N^\epsilon$.
\end{rem}

The proof of Proposition~\ref{le main proposition} is postponed to Section~\ref{section proof of proposition} and we next show how it implies Theorem~\ref{le main theorem}.

\begin{proof}[Proof of Theorem~\ref{le main theorem}]
 In order to link~\eqref{le claim} to~\eqref{le main theorem equation} we observe that  by spectral decomposition we have
\begin{align}\label{le spectral decomposition}
 \frac{1}{N}\text{Tr} \mathcal{H} \Im G(z_1)\mathcal{H} \Im G(z_2)=\frac{1}{N}\sum_{\alpha,\beta=1}^N| w_\alpha^*\mathcal{H} w_\beta|^2\frac{\eta_1}{(\lambda_\alpha-E_1)^2+\eta_1^2}\frac{\eta_2}{(\lambda_\beta-E_2)^2+\eta_2^2}\,,
\end{align}
where $z_1=E_1+\ii\eta_1$, $z_2=E_2+\ii\eta_2$, $\eta_1\not=0$, $\eta_2\not=0$.

 Fix now indices $\alpha,\beta$ and choose $E_1=\lambda_\alpha$ and $E_2=\lambda_\beta$, as well as $\eta_1=\eta_2=N^{-1+\epsilon}$ such that~$z_1,z_2\in\mathcal{E}$ with very high probability by~\eqref{le rigidity estimate}. Then we obtain from the uniform bound in~\eqref{le claim} combined with the representation~\eqref{le spectral decomposition} the estimate
\begin{align}\label{le matrix element}
 | w_\alpha^*\mathcal{H} w_\beta|^2\prec N\eta_1\eta_2\prec N^{-1}\,,
\end{align}
for all $\alpha,\beta\in\llbracket 1,N\rrbracket$.

Next, similarly to~\eqref{le baby computation}, we conclude by noticing that
\begin{align*}
  w_\alpha^*H_1 w_\beta-\sigma_1\lambda_\alpha\delta_{\alpha\beta}&= w_\alpha^*H_1 w_\beta-\sigma_1^2 w_\alpha^* H_1w_\beta-\sigma_1\sigma_2w_\alpha^*H_2w_\beta\nonumber\\
 &=\sigma_2^2 w_\alpha^*H_1 w_\beta-\sigma_1\sigma_2w_\alpha^*H_2w_\beta\nonumber\\
 &=\sigma_2 w_\alpha^*\mathcal{H}w_\beta\nonumber\\
 &=O_\prec\Big(\frac{1}{\sqrt{N}}\Big)\,,
\end{align*}
where we used~\eqref{le matrix element}. This concludes the proof of Theorem~\ref{le main theorem}.
\end{proof}

 \section{Computation of the expectation}\label{section expectation}

\newcommand{\caD}{\mathcal{D}}
\newcommand{\caX}{\mathcal{X}} 
In this section we compute the expectation of the observable $\langle\mathcal{H} \im G(z_1)\mathcal{H}  \im G(z_2)\rangle$. Since this random variable is, for $z_1,z_2\in\C^+$, positive, the expectation already indicates its correct size. Also the estimation of the expectation unveils the cancellation mechanism Theorem~\ref{le main theorem} eventually results from.

\begin{lem}\label{le lemma expectation final}
 Let $z_1,z_2\in\mathcal{E}$. Then,
 \begin{align}
  \E\,\langle \mathcal{H}\im G(z_1)\mathcal{H} \im G(z_2)\rangle=\im m_{sc}(z_1)\im m_{sc}(z_2)+O_\prec\Big(\frac{1}{\sqrt{N}}\Big)+O_\prec\Big(\Psi^2(z_1,z_2)\Big)\,,
 \end{align}
where $\Psi(z_1,z_2)$ is defined in~\eqref{le Psi}.
\end{lem}
\begin{proof}
We start by noticing that it suffices to estimate 
\begin{align}\label{le X}
 \mathcal{X}(z,z'):=\langle \caH G(z)\caH G(z')\rangle\,.
\end{align}
for $z=z_1, \overline{z}_1$ and $z'=z_2, \overline{z}_2$. Further introduce the short hand notation
 \begin{align}\label{le G and F}
  G\equiv G(z)\,,\qquad G'\equiv G(z')\,.
 \end{align}
Moreover, note that we can write
 \begin{align}\label{le for sum convetion}
 \mathcal{X}(z,z')=\frac{1}{N}\sum_{ijab}\caH_{ij}\Tr (\Delta^{ij}G\Delta^{ab}G')\caH_{ab}\,,
 \end{align}
where the matrix $(\Delta^{ij})$ is defined to have entries $(\Delta^{ij})_{nm}=\delta_{in}\delta_{jm}$, or using rank one operators 
\begin{align}\label{le Delta}\Delta^{ij}:=|e_i\rangle \langle e_j|\,,\end{align}
where $(e_i)_i$ is the canonical basis in $\C^N$. Recall that $\sum_{ijab}$ indicates
a sum over all indices from $1$ to $N$.

 Our task is to compute
\begin{align}\label{le mesoscopic written out}
 \E\mathcal{X}(z,z')&=\frac{1}{N}\sum_{ijab} \E \big[\caH_{ij}\Tr(\Delta^{ij}G\Delta^{ab}G')\caH_{ab}\big]\,.
\end{align}
For this we use the cumulant expansions of Lemma~\ref{complex_cumulant}. To get started, we need  more notation. Let $\kappa_{\iota,ij}^{(p,q)}=\overline{\kappa_{\iota,ji}^{(p,q)}}$ denote the cumulants of the matrix entries $h_{\iota,ij}$, $\iota=1,2$. We will for simplicity assume for the moment that $\E h_{\iota,ij}^2=0$, this condition can easily be relaxed; see Section~\ref{remark about variance complex case}. Together with Assumption~\ref{assumption 1} this implies
\begin{align}\label{the first few kappas}
 \kappa_{\iota,{ij}}^{(1,0)}&=\kappa_{\iota,{ij}}^{(0,1)}=0\,,\qquad\kappa_{\iota,{ij}}^{(1,1)}=\frac{1}{N}\,,\qquad \kappa_{\iota,{ij}}^{(2,0)}=\kappa_{\iota,{ij}}^{(0,2)}=0\,.
\end{align}
Further, from~\eqref{entries moment bounds} in Assumption~\ref{assumption 1} we have the estimates
\begin{align}\label{le size of kappa}
 |\kappa_{\iota,ij}^{(p,q)}|\le\frac{C_{p+q}}{N^{\frac{p+q}{2}}}\,,\qquad p+q\ge 3\,.
\end{align}

Next, introduce the derivation operator
\begin{align}\label{le D}
 \caD_{ji}:=(\sigma_{2}\partial_{1,ji}-\sigma_{1}\partial_{2,ji})\,,
\end{align}
where $\partial_{\iota,ji}\equiv \frac{\partial}{\partial h_{\iota,ji}}$, $\iota=1,2$.  

We now have the computational rules,

\begin{align}\label{rule 1}
 \caD_{ji}\caH_{ab}=\delta_{ja}\delta_{ib}\,,
\end{align}
where we used~\eqref{le sigma condition}, and
\begin{align}\label{rule 2}
 \caD_{ji}G(z)=-\sigma_{2}\sigma_{1}G(z)\Delta^{ji} G(z)+\sigma_{1}\sigma_{2}G(z)\Delta^{ji} G(z)=0\,,
\end{align}
where we used the basic differential rule
\begin{align}\label{le basic rule derivative}
 \partial_{\iota,ji}G(z)=-G(z)\sigma_{\iota}\Delta^{ji}G(z)\,,\qquad\quad \iota=1,2\,.
\end{align}

We will also require a higher order analogue of $\mathcal{D}$: For $p,q\in\N$ define
\begin{align}\label{le fancy D}
 \caD_{ji}^{(p,q)}:= \frac{1}{p!q!}N^{\frac{p+q+1}{2}}\Big(\sigma_{2}\kappa_{1,ji}^{(p,q+1)}\partial_{1,ji}^p\partial_{1,ij}^q-\sigma_{1}\kappa_{2,ji}^{(p,q+1)}\partial_{2,ji}^p\partial_{2,ij}^q\Big)\,,
 \end{align}
with this notation we have $\caD_{ji}=\caD_{ji}^{(1,0)}$ and record that 
\begin{align}\label{rule 2bis}
\caD_{ji}^{(1,0)}G=0\,,\qquad\caD_{ji}^{(0,1)}=0\,,
\end{align}
where the first relation follows from~\eqref{rule 2}, while the second follows from $\kappa^{(0,2)}=0$; see~\eqref{the first few kappas}. 
With the notation in~\eqref{le fancy D} we next recall Lemma~\ref{complex_cumulant} to obtain the following cumulant expansion lemma.

\begin{lem}\label{lemma cumulant expansion}
 Fix indices $i,j$ and integers $d,d'$. Let $F$ be a monomial in the Green function entries $(G_{nm}(z))_{nm}$, $(G_{nm}(z'))_{nm}$, and matrix entries $(\mathcal{H}_{nm})_{nm}$ of total degree $d$ in the Green function entries and total degree $d'$ in $\mathcal{H}_{nm}$ where $d'\le d$. Then for any fixed $l\in\N$, 
 \begin{align}\label{le second cumulant expansion1}
  \E_{ij} \caH_{ij} F=\sum_{p+q=1}^l\frac{1}{N^{\frac{p+q+1}{2}}} \E_{ij}\mathcal{D}_{ji}^{(p,q)}F+\Omega_{l+1}(F)\,,
 \end{align}
where $\E_{ij}$ denotes the expectation with respect to the random variables $h_{1,ij}$ and $h_{2,ij}$. The error term satisfies the bound
\begin{align}
|\Omega_{l+1}(F)|\prec N^{-(l+2)/2}\,,
\end{align}
where the explicit constants depend on $d$ and $d'$, but are uniform in the matrix indices.

\end{lem}
The proof of Lemma~\ref{lemma cumulant expansion} is postponed to Appendix~\ref{le appendix}. Lemma~\ref{lemma cumulant expansion} has the following direct corollary whose proof is postponed to Appendix~\ref{le appendix}, too.
 \begin{cor}\label{cor cumulant expansion}
 Fix indices $i,j$. Let $F$ be a monomial in the Green function entries $(G_{nm}(z))_{nm}$, $(G_{nm}(z'))_{nm}$, and matrix entries $(\mathcal{H}_{nm})_{nm}$ of total degree $d$ in Green function entries and total degree $d'$ in $\mathcal{H}_{nm}$ where $d'\le d$. Then for any fixed $l\in\N$, 
 \begin{align}\label{le second cumulant expansion}
  \E \caH_{ij} F=\sum_{p+q=1}^l\frac{1}{N^{\frac{p+q+1}{2}}} \E \mathcal{D}_{ji}^{(p,q)}F+\E\Omega_{l+1}(F)\,,
 \end{align}
where the error term satisfies the bound
\begin{align}
|\E \Omega_{l+1}(F)|\prec N^{-(l+2)/2}\,,
\end{align}
where the explicit constants depend on $d$ and $d'$, but are uniform in the matrix indices.

\end{cor}

With Corollary~\ref{cor cumulant expansion} and the computational rules~\eqref{rule 1} and~\eqref{rule 2bis} at hand, we begin to compute the expectation of $\caX(z,z')$:
\begin{align}\label{le expansion of expectation}
 \E\caX(z,z')&=\frac{1}{N}\sum_{ijab} \E \caH_{ij}\Tr(\Delta^{ij}G\Delta^{ab}G')\caH_{ab}\nonumber\\
 &=\frac{1}{N}\sum_{p+q=1}^l\frac{1}{N^{\frac{p+q+1}{2}}} \sum_{ijab}\E\mathcal{D}_{ji}^{(p,q)}\Big[\Tr(\Delta^{ij}G\Delta^{ab}G')\caH_{ab}\Big]+O_\prec(N^{(-l+4)/2})\nonumber\\
 &=\frac{1}{N^2} \sum_{ijab}\E\mathcal{D}_{ji}^{(1,0)}\Big[\Tr(\Delta^{ij}G\Delta^{ab}G')\caH_{ab}\Big]\nonumber\\ &\qquad+\frac{1}{N}\sum_{p+q=2}^l\frac{1}{N^{\frac{p+q+1}{2}}} \sum_{ijab}\E\mathcal{D}_{ji}^{(p,q)}\Big[\Tr(\Delta^{ij}G\Delta^{ab}G')\caH_{ab}\Big]+O_\prec(N^{(-l+4)/2})\,,
\end{align}
where we used Corollary~\ref{cor cumulant expansion} together with~\eqref{le size of kappa} and power counting to estimate the error term from cutting the cumulant expansion at order $l$ to be $N^{-1}N^{4}O_\prec(N^{-(l+2)/2})=O_\prec(N^{(-l+4)/2})$.

We first focus on the first term on the right side of~\eqref{le expansion of expectation}. Using~\eqref{rule 1} and~\eqref{rule 2}, we get
\begin{align}\label{le needed in real case too}
 \frac{1}{N^2} \sum_{ijab}\E\mathcal{D}_{ji}^{(1,0)}\Big[\Tr(\Delta^{ij}G\Delta^{ab}G')\caH_{ab}\Big]&
 =\frac{1}{N^2}\sum_{ijab}\E\Tr (\Delta^{ij}G\Delta^{ab}G')\delta_{ja}\delta_{ib}\nonumber\\
 &\qquad+\frac{1}{N^2}\sum_{ijab}\E\Tr (\Delta^{ij}(\caD_{ji}G)\Delta^{ab}G')\caH_{ab}\nonumber\\ &\qquad+\frac{1}{N^2}\sum_{ijab}\E\Tr (\Delta^{ij}G\Delta^{ab}(\caD_{ji}G'))\caH_{ab}\nonumber\\
 &=\frac{1}{N^2}\sum_{ij}\E G_{jj}G_{ii}'=\E m(z)m(z')\,,
\end{align}
Note that the only non-zero term is when $\caD_{ji}$ acts on $\caH_{ab}$. By the local law in~\eqref{le local law}, and the deterministic estimate $|m(z)|\le\frac{1}{|\eta|}\le N$ together with item $(3)$ of Lemma~\ref{dominant}, the first term on the right side of~\eqref{le expansion of expectation} is thus given by
\begin{align}\label{le contribution 1}
  \frac{1}{N^2} \sum_{ijab}\E\mathcal{D}_{ji}^{(1,0)}\Big[\Tr(\Delta^{ij}G\Delta^{ab}G')\caH_{ab}\Big]&=m_{sc}(z)m_{sc}(z')+O_\prec\big(\Psi(z,z')^2)\,.
\end{align}

Consider next the second term on the right of~\eqref{le expansion of expectation}. We are going to use yet another cumulant expansion with respect to $\caH_{ab}$ to exploit further cancellation based on~\eqref{rule 2bis}. For this purpose we first note that if $\{a,b\}\not=\{i,j\}$ as sets, then
\begin{align}\label{le switch derivata}
 \mathcal{D}_{ji}^{(p,q)}\big(\Tr(\Delta^{ij}G\Delta^{ab}G')\caH_{ab}\big)=\caH_{ab}\mathcal{D}_{ji}^{(p,q)}\big(\Tr(\Delta^{ij}G\Delta^{ab}G')\big)\,,
\end{align}
because then $\partial_{1,ij}\caH_{ab}=\partial_{2,ij}\caH_{ab}=0$. If $\{a,b\}=\{i,j\}$, then by power counting using~$|\caH_{ab}|\prec1$ and the boundedness of the Green function entries, we can estimate
 \begin{align*}
 &\Big|\frac{1}{N}\sum_{p+q=2}^l\frac{1}{N^{\frac{p+q+1}{2}}} \sum_{ijab} \lone_{(\{a,b\}=\{i,j\})}\E\mathcal{D}_{ji}^{(p,q)}\Big[\Tr(\Delta^{ij}G\Delta^{ab}G')\caH_{ab}\Big]\Big|\prec\frac{1}{N}\frac{1}{N^{3/2}}N^2\prec\frac{1}{\sqrt{N}}\,,
\end{align*}
where we tacitly used item $(3)$ of Lemma~\ref{dominant}, together with H\"older's inequality and the deterministic estimate $\|G(z)\|\le|\eta|^{-1}\le N$ and the moment bounds in~\eqref{entries moment bounds}.
Hence, we have for the second term on the right side of~\eqref{le expansion of expectation} that
\begin{align}\label{le contribution 1bis}
 &\frac{1}{N}\sum_{p+q=2}^l\frac{1}{N^{\frac{p+q+1}{2}}} \sum_{ijab}\E\mathcal{D}_{ji}^{(p,q)}\Tr(\Delta^{ij}G\Delta^{ab}G')\caH_{ab}
 \nonumber\\&\qquad\qquad=\frac{1}{N}\sum_{p+q=2}^l\frac{1}{N^{\frac{p+q+1}{2}}} \sum_{\{i,j\}\not=\{a,b\}}\E\caH_{ab}\mathcal{D}_{ji}^{(p,q)}\Tr(\Delta^{ij}G\Delta^{ab}G')+O_\prec\big(\frac{1}{\sqrt{N}}\big)\,.
\end{align}

Next, using a cumulant expansion to order $l$ with respect to $\mathcal{H}_{ab}$, we get
\begin{align*}
 &\frac{1}{N}\sum_{p+q=2}^l\frac{1}{N^{\frac{p+q+1}{2}}} \sum_{\{i,j\}\not=\{a,b\}}\E\mathcal{D}_{ji}^{(p_1,q_1)}\Tr(\Delta^{ij}G\Delta^{ab}G')\caH_{ab}\nonumber\\ 
 &\qquad= \frac{1}{N}\sum_{\substack{p_1+q_1=2 \\ p_2+q_2=1}}^l\frac{1}{N^{\frac{p_1+q_1+p_2+q_2+2}{2}}} \sum_{\{i,j\}\not=\{a,b\}}\E\mathcal{D}_{ba}^{(p_2,q_2)}\mathcal{D}_{ji}^{(p_1,q_1)}\Tr(\Delta^{ij}G\Delta^{ab}G')+O_\prec(N^{(-l+1)/2})\,.
\end{align*}
By~\eqref{rule 2bis}, we see that the terms with $p_2+q_2=1$ yield a zero contribution, so we have
\begin{align}\label{le one intermediate step}
 &\frac{1}{N}\sum_{p+q=2}^l\frac{1}{N^{\frac{p+q+1}{2}}} \sum_{\{i,j\}\not=\{a,b\}}\E\mathcal{D}_{ji}^{(p,q)}\Tr(\Delta^{ij}G\Delta^{ab}G')\caH_{ab}\nonumber\\ 
 &\qquad= \frac{1}{N}\sum_{\substack{p_1+q_1=2\\ p_2+q_2=2}}^l\frac{1}{N^{\frac{p_1+q_1+p_2+q_2+2}{2}}} \sum_{\{i,j\}\not=\{a,b\}}\E\mathcal{D}_{ba}^{(p_2,q_2)}\mathcal{D}_{ji}^{(p_1,q_1)}\Tr(\Delta^{ij}G\Delta^{ab}G')+O_\prec(N^{(-l+1)/2})\nonumber\\
 &\qquad= \frac{1}{N^4}\sum_{\{i,j\}\not=\{a,b\}}\E\mathcal{D}_{ba}^{(1,1)}\mathcal{D}_{ji}^{(1,1)}\Tr(\Delta^{ij}G\Delta^{ab}G')\nonumber\\
 &\qquad\qquad+\frac{1}{N}\sum_{\substack{p_1+q_1+p_2+q_2\ge5}}^l\frac{1}{N^{\frac{p_1+q_1+p_2+q_2+2}{2}}} \sum_{\{i,j\}\not=\{a,b\}}\E\mathcal{D}_{ba}^{(p_2,q_2)}\mathcal{D}_{ji}^{(p_1,q_1)}\Tr(\Delta^{ij}G\Delta^{ab}G')\nonumber\\ &\qquad\qquad+O_\prec(N^{(-l+1)/2})\,.
\end{align}
Using the local law for the Green function entries in~\eqref{le local law} and Lemma~\ref{dominant}, we can easily bound the second term on the right side by
\begin{align}\label{le contribution 3}
 \Big|\frac{1}{N}\sum_{\substack{p_1+q_1+p_2+q_2\ge5}}^l\frac{1}{N^{\frac{p_1+q_1+p_2+q_2+2}{2}}} \sum_{\{i,j\}\not=\{a,b\}}\E\mathcal{D}_{ba}^{(p_2,q_2)}\mathcal{D}_{ji}^{(p_1,q_1)}\Tr(\Delta^{ij}G\Delta^{ab}G')\Big|\prec \frac{1}{\sqrt{N}}\,.
\end{align}

For the first term on the right side of~\eqref{le one intermediate step}, we observe that $\mathcal{D}_{ba}^{(1,1)}\mathcal{D}_{ji}^{(1,1)}$ contains four partial derivatives. When those act on the Green function entries $\Tr(\Delta^{ij}G\Delta^{ab}G')= G_{ja}G'_{bi}$ they create  by~\eqref{le basic rule derivative} monomials of degree six in the Green function entries. Assuming that $a,b,i,j$ are all distinct, the four partial derivatives will create diagonal as well as off-diagonal Green function entries when acting on $G_{ja}G'_{bi}$ since, e.g.\ $\partial_{1,ba}G_{ja}=-\sigma_{1}G_{jb}G_{aa}$. Note that the total number of  off-diagonal entries does not decrease, hence each resulting monomial contains at least two off-diagonal entries. In power counting we count diagonal entries as $O_\prec(1)$ while the off-diagonal are counted as $O_\prec(\Psi)$. If there are coincidences among the indices, we gain a factor $1/N$ in the summation for each coincidence, hence those are negligible when compared with $\Psi^2$. We hence have the estimate

\begin{align}\label{le contribution 2}
 \Big|\frac{1}{N^4}\sum_{\{i,j\}\not=\{a,b\}}\E\mathcal{D}_{ba}^{(1,1)}\mathcal{D}_{ji}^{(1,1)}\Tr(\Delta^{ij}G\Delta^{ab}G')\Big|\prec (\Psi(z,z'))^2\,.
\end{align}

In sum, choosing $l\ge 5$, we get from~\eqref{le contribution 2}, ~\eqref{le contribution 3},~\eqref{le contribution 1bis}, and~\eqref{le contribution 1} that
\begin{align}\label{le hws}
\E\mathcal{X}(z,z')=\frac{1}{N}\E \Tr \caH G(z)\caH G(z')=m_{sc}(z)m_{sc}(z')+O_\prec\big((\Psi(z,z'))^2\big)+O_\prec(\frac{1}{\sqrt{N}})\,.
\end{align}
Using linear combinations, Lemma~\ref{le lemma expectation final} follows directly from~\eqref{le hws}.
\end{proof}

\section{Proof of Proposition~\ref{le main proposition}}\label{section proof of proposition}

In the previous section we identified the expectation of $\langle \caH G\caH G'\rangle$ in Lemma~\ref{le lemma expectation final}. In the current section, we will control the higher moments of $\langle \caH G\caH G'\rangle$ to obtain a high probability bound required to prove Proposition~\ref{le main proposition}.

\begin{pro}\label{le proposition}
Under the assumptions of Theorem~\ref{le main theorem}, we have
\begin{align}\label{le proposition equation}
 \caX(z,z')=m(z)m(z')+O_{\prec}(1)\,,
\end{align}
uniformly in $z,z'\in\mathcal{E}$.
\end{pro}

\begin{proof}
 We rewrite $\caX$ as
\begin{align}
 \mathcal{X}(z,z')=\frac1N\sum_{ijab} \caH_{ij}\Tr(\Delta^{ij}G\Delta^{ab}G')\caH_{ab}=:\frac1N\sum_{ijab} \caH_{ij}X^{ijab}	\caH_{ab}\,,
\end{align}
where we introduced
\begin{align}\label{le funny X}
 X^{ijab}\equiv \Tr(\Delta^{ij}G\Delta^{ab}G')=G_{ja}G'_{bi}\,.
\end{align}
As in Section~\ref{section expectation}, we assume for the moment that $\E h_{\iota,ij}^2=0$, $\iota=1,2$, $i,j\in\llbracket1,N\rrbracket$. This implies that $\kappa_{\iota,ij}^{(0,2)}=\kappa_{\iota,ij}^{(2,0)}=0$ as well as $\mathcal{D}_{ji}^{(0,1)}\equiv 0$. We are going to explain in Section~\ref{remark about variance complex case} how this additional assumption can easily be dropped. 

Next, we observe from~\eqref{rule 1} and~\eqref{rule 2bis} that
\begin{align}\label{le kappa rule}
\frac{1}{N^2} \sum_{ijab}\mathcal{D}_{ji}^{(1,0)}X^{ijab}\caH_{ab}=\frac{1}{N^2}\sum_{ij} G_{ii}G_{jj}'=m(z)m(z'):=\varkappa(z,z')\,,
\end{align}
where we introduce the shorthand $\varkappa$. For $n,m\in\N$, define
\begin{align}\label{le P recursive}
 P(n,m):=(\caX-\varkappa)^n\overline{(\caX-\varkappa)}^m\,.	
\end{align}

Fix a (large) $D\in\N$, $z,z'\in\mathcal{E}$, and consider
\begin{align}\label{le monster}
 \E P(D,D)&=\sum_{\mathbf{i}\mathbf{j}\mathbf{a}\mathbf{b}}\E\Bigg[\prod_{n=1}^D\Big(\frac{1}{N}\caH_{i_{n}j_{n}}X^{i_{n}j_{n}a_{n}b_{n}}\caH_{a_{n}b_{n}}-\frac{1}{N^2}\delta_{j_na_n}\delta_{i_nb_n}G_{i_ni_n}G_{j_nj_n}\Big)\nonumber\\ &\qquad\qquad\times\prod_{n=D+1}^{2D}\Big(\frac{1}{N}\caH_{j_{n}i_{n}}\overline{X^{i_{n}j_{n}a_{n}b_{n}}}\caH_{b_{n}a_{n}}-\frac{1}{N^2}\delta_{j_na_n}\delta_{i_nb_n}\overline{G_{i_ni_n}}\,{\overline{G_{j_nj_n}}}\Big)\Bigg]\,,
\end{align}
where $\mathbf{i}=(i_1,i_2,\ldots, i_{2D})\in\llbracket 1,N\rrbracket ^{2D}$, and similarly $\mathbf{j},\mathbf{a},\mathbf{b}\in\llbracket 1,N\rrbracket^{2D}$ are $8D$ free summation indices corresponding to $4D$ factors of $\caH$'s. In the expression above, we call, for each $n$, $\caH_{i_nj_n}$ and $\caH_{a_nb_n}$ `twins'.

We now successively use the cumulant expansions from Corollary~\ref{cor cumulant expansion} to expand the summands in~\eqref{le monster} in all the factors of $\caH$'s. We start by expanding in the variable $\mathcal{H}_{i_1j_1}$ to obtain
\begin{align}\label{le 6.8}
 &\E\Big[\big(\caX(z,z')-\varkappa(z,z')\big)P(D-1,D)\Big]\nonumber\\&\qquad=\frac{1}{N}\sum_{p_1+q_1=1}^l\sum_{i_1j_1a_1b_1}\frac{1}{N^{\frac{p_1+q_1+1}{2}}} \E\bigg[ \mathcal{D}_{j_1i_1}^{(p_1,q_1)}\Big[X^{i_1j_1a_1b_1}\caH_{a_1b_1} P(D-1,D)\Big]\bigg]\nonumber\\
 &\qquad\qquad-\E\Big[ \varkappa P(D-1,D)\Big]+\E\Big[O_\prec(N^{\frac{-l+4}{2}})P(D-1,D)\Big]\,.
\end{align}
First, using that $|\caX|\le \eta_0^{-2}\le N^2$ and $|\varkappa|\le\eta_0^{-2}\le N^2$, with $\eta_0=\min\{|\Im z|,\,|\Im z'|\}$, the third term on the right of~\eqref{le 6.8} is bounded as $O_\prec(N^{\frac{-l+4}{2}+(4D-2)})$, hence for $l\ge 10D$, that error term is bounded as $O_\prec\big((\frac{1}{\sqrt{N}})^{2D}\big)$. Here, we also tacitly used, as we will do repeatedly below, item $(3)$ of Lemma~\ref{dominant} to justify the estimate. Second, in the first term on the right, for $p+q=1$, we consider the derivation $\caD^{(1,0)}_{j_1i_1}$ (recall from~\eqref{the first few kappas} that $\caD^{(0,1)}_{j_1i_1}=0$). When $\caD^{(1,0)}_{j_1i_1}$ acts on a Green function in $X^{i_1j_1a_1b_1}$ we get a zero contribution thanks to~\eqref{rule 2bis}. If $\caD^{(1,0)}_{j_1i_1}$ acts on its twin $\caH_{a_1b_1}$ we generate by~\eqref{le kappa rule} the term $\E[ \varkappa P(D-1,D)]$ which will precisely cancel with the second term on the right side of~\eqref{le 6.8}. 

Thus, choosing $l\ge 10D$, we have
\begin{align}\label{le it aint me}
\E P(D,D)&=\frac{1}{N^2}\sum_{i_1j_1a_1b_1}\E\Big[X^{i_1j_1a_1b_1}\caH_{a_1b_1} \mathcal{D}_{j_1i_1}^{(1,0)} P(D-1,D)\Big]\nonumber\\
&\qquad+\frac{1}{N}\sum_{p_1+q_1=2}^l\sum_{i_1j_1a_1b_1}\frac{1}{N^{\frac{p_1+q_1+1}{2}}} \E\bigg[ \mathcal{D}_{j_1i_1}^{(p_1,q_1)}\Big[X^{i_1j_1a_1b_1}\caH_{a_1b_1} P(D-1,D)\Big]\bigg]\nonumber\\
 &\qquad+O_\prec\Big(\big(\frac{1}{\sqrt{N}}\big)^{2D}\Big)\,.
\end{align}

 Consider now the first term on the right side of~\eqref{le it aint me}. When $\mathcal{D}_{j_1i_1}^{(1,0)}$ acts on $P(D-1,D)$ it either acts on a Green function entry $G_{i_nj_n}$ or $G_{a_nb_n}$,  or it acts on $\mathcal{H}_{i_nj_n}$ or $\caH_{a_nb_n}$, $n\in \llbracket 2,\ldots, 2D\rrbracket$. In the former case we get by~\eqref{rule 2bis} a zero contribution, in the latter case by~\eqref{rule 1} the number of free summation indices in $P(D-1,D)$ gets reduced from $4(2D-1)$ to $4(2D-1)-2$. Bearing this in mind, we expand the first term on the right side of~\eqref{le it aint me} using $\caH_{a_1b_1}$ to obtain, with $l\ge 10D$,
 \begin{align}\label{le eng2}
 &\frac{1}{N^2}\sum_{i_1j_1a_1b_1}\E\Big[X^{i_1j_1a_1b_1}\caH_{a_1b_1} \mathcal{D}_{j_1i_1}^{(1,0)} P(D-1,D)\Big]\nonumber\\&\qquad=\frac{1}{N^2}\sum_{i_1j_1a_1b_1}\sum_{p_2+q_2=1}^l\frac{1}{N^{\frac{p_2+q_2+1}{2}}}\E\bigg[\mathcal{D}_{b_1a_1}^{(p_2,q_2)}\Big[X^{i_1j_1a_1b_1}\mathcal{D}_{j_1i_1}^{(1,0)} P(D-1,D)\Big]\bigg]+O_\prec\Big(\big(\frac{1}{\sqrt{N}}\big)^{2D}\Big)\,,\nonumber\\
 &\qquad=\frac{1}{N^2}\sum_{i_1j_1a_1b_1}\frac{1}{N}\E\Big[X^{i_1j_1a_1b_1}\mathcal{D}_{b_1a_1}^{(1,0)}\mathcal{D}_{j_1i_1}^{(1,0)} P(D-1,D)\Big]\nonumber\\ &\qquad\qquad+\frac{1}{N^2}\sum_{i_1j_1a_1b_1}\sum_{p_2+q_2=2}^l\frac{1}{N^{\frac{p_2+q_2+1}{2}}}\E\bigg[\mathcal{D}_{b_1a_1}^{(p_2,q_2)}\Big[X^{i_1j_1a_1b_1}\mathcal{D}_{j_1i_1}^{(1,0)} P(D-1,D)\Big]\bigg]\nonumber\\&\qquad\qquad+O_\prec\Big(\big(\frac{1}{\sqrt{N}}\big)^{2D}\Big)\,,
 \end{align}
where we used~\eqref{rule 2bis}. For the first term on the right side, the number of free summation indices in $\mathcal{D}_{b_1a_1}^{(1,0)}\mathcal{D}_{j_1i_1}^{(1,0)} P(D-1,D)$ is  $4(2D-1)-4$ by~\eqref{rule 1} and~\eqref{rule 2bis}. Or put differently, there are $2(2D-1)-2$ factors of $\caH$'s left that we can use in cumulant expansions.

For the second term on the right side of~\eqref{le eng2}, we either get a zero contribution when $\mathcal{D}_{j_1i_1}^{(1,0)}$ acts on a Green function entry of $P(D-1,D)$, or the number of free summation indices gets reduced by two if $\mathcal{D}_{j_1i_1}^{(1,0)}$ acts on a factor of $\caH$. For the higher derivative terms in~$\mathcal{D}_{b_1a_1}^{(p_2,q_2)}$, with $p_2+q_2\ge 2$, acting on $X^{i_1j_1a_1b_1}\mathcal{D}_{j_1i_1}^{(1,0)} P(D-1,D)$, either the number of Green function entries is increased by one for each derivative hitting a Green function entry, or the number of free summation indices is reduced by two for each derivative hitting a factor $\caH$. We have now expanded the first term on the right of~\eqref{le it aint me} in $\caH_{i_1j_1}$ and $\caH_{a_1b_1}$. Before we go on and expand the remaining $\caH$'s in $P(D-1,D)$, we return to second term on the right of~\eqref{le it aint me}.

Consider the second term on the right side of~\eqref{le it aint me}. Since $p_1+q_1\ge 2$, we do not have further cancellations from~\eqref{rule 2bis} in $\mathcal{D}_{j_1i_1}^{(p_1,q_1)}X^{i_1j_1a_1b_1}\caH_{a_1b_1} P(D-1,D) $. If one of the derivatives in~$\mathcal{D}_{j_1i_1}^{(p_1,q_1)}$ acts on $\caH_{a_1b_1}$, the number of free summation indices is reduced by two, if none of the derivatives act on $\caH_{a_1b_1}$, we use a cumulant expansion in $\caH_{a_1b_1}$ stopped at order $l\ge 10D$. The leading term containing $\caD_{b_1a_1}^{(1,0)}$ will then either give a zero contribution if it acts on any Green function entry by~\eqref{rule 2bis} or it will reduce the number of free summation indices by two. For the terms containing~$\caD_{b_1a_1}^{(p_2,q_2)}$, $p_2+q_2\ge 2$, we have no cancellation due to~\eqref{rule 2bis} but the number of free summation indices gets reduced by two for each derivative acting on a factor $\caH$.

To sum up, after performing all the derivatives by Leibniz rule, the terms on the right side of~\eqref{le it aint me} can be classified by the number of collapses, $M$, of two free summation indices when $\caH_{i_1j_1}$ or $\caH_{a_1b_1}$ act on some other $\caH$'s (except their own twin), and the number of cumulant expansions $L$ in total; the number of cumulant expansions, $L_1$, starting from order one, i.e.\ with $p_n+q_n\ge 1$; and the number of cumulant expansions, $L_2$, starting from order two, i.e.\ with $p_n+q_n\ge 2$. For the moment either $L=1$ or $2$, with $L_1+L_2=L$.  Because of the bounds $|G_{ij}(z)|\prec 1$, $\|G(z)\|\le|\eta|^{-1}\le N$ and Lemma~\ref{dominant}$(3)$, we may ignore the number of Green function entries in the power counting and do not keep track of them.

We have now fully expanded~\eqref{le it aint me} in terms of $\caH_{i_1j_1}$ and $\caH_{a_1b_1}$. We will continue expanding in the remaining $\caH$'s while keeping track of the numbers $M$, $L_1$ and $L_2$ introduced above.

Pick now one of the resulting terms from above, if that term contains $\caH_{i_2j_2}$ and its twin $\caH_{a_2b_2}$ we expand first in $\caH_{i_2j_2}$. When $\caD^{(0,1)}_{b_2a_2}$ acts on $\caH_{a_2b_2}$ we get the cancellation with $\varkappa$ from~\eqref{le kappa rule}, so that we are left with a cumulant expansion with $p_2+q_2\ge 2$ only. In case the twin $\caH_{a_2b_2}$ is missing, we note that the number of free summation indices has already been reduced by two. If we pick a term that does not contain $\caH_{i_2j_2}$, we go on and expand in the next $\caH$, $\caH_{a_2b_2}$ or if missing the next available $\caH$. In this way we successively expand all factors $\caH$'s, except those appearing in the error term of a cumulant expansion cut at order $l\ge 10D$.

A resulting fully expanded term containing no more $\caH$'s is then classified by the total number of collapses of free summation indices, $M$, resulting from~\eqref{rule 1}. The number of free summation indices in such a term is $8D-2M$ whereas the number of total cumulant expansion, $L$, in that term is $4D-M$. As above, let $L_1$ be the number of cumulant expansions with $p_n+q_n= 1$ and let $L_2$ be the number of cumulant expansions with $p_n+q_n\ge 2$. Note that $L_2=4D-M-L_1$.

Hence a fully expanded term with given $M$, $L_1$ and $L_2$ gives a contribution to~\eqref{le monster} bounded~by
\begin{align}\label{the big outcome}
&\frac{1}{N^{2D}}N^{8D-2M}\Big[\sum_{p_n+q_n=1}^l\Big(\frac{1}{\sqrt{N}}\Big)^{p_n+q_n+1}\Big]^{L_1} \Big[\sum_{p_n+q_n=2}^l\Big(\frac{1}{\sqrt{N}}\Big)^{p_n+q_n+1}\Big]^{L_2}\nonumber\\
&\qquad\prec N^{6-2M}N^{-L_1}N^{-\frac32 L_2}\nonumber\\
&\qquad=\Big(\frac{1}{\sqrt{N}}\Big)^{-M+L_1}\,,
\end{align}
where we used that $|G_{ij}|\prec 1$, as well as $\|G(z)\|\le \frac{1}{|\eta|}\le N$ with probability one and that there are no more $\caH$'s in a fully expanded term so that by Lemma~\ref{dominant} we get the first line. To obtain the second line we used that $D$ and $l\ge 10D$ are fixed numbers, and for the third line we used that $L_2=4D-M-L_1$. Summarizing, so far we have expanded~\eqref{le monster} in all the factors $\caH$ and showed that each resulting fully expanded term with given $M$, $L_1$ and $L_2$ is bounded by~\eqref{the big outcome}.

 We next claim that $L_1\le M$ for any fully expanded term. Indeed if for some pair of indices $i_nj_n$ or $a_nb_n$ there is no collapse, meaning that the derivatives in $\caD^{(p_n,q_n)}_{j_ni_n}$ (or $\caD^{(p_n,q_n)}_{b_na_n}$) exclusively acted on Green function entries, then we have due to~\eqref{rule 2bis} that $p_n+q_n\ge 2$ in order to get a non-zero contribution. 
 
 Thus we reach the maximum for $M=L_1$ in~\eqref{the big outcome}, and the term is stochastically dominated by one, i.e.\ each fully expanded term is stochastically bounded by one. The number of generated terms in the expansion is bounded by $(CD)^{cD}$ if we choose $l$ to be proportional to $D$. 
 
 It follows that 
 \begin{align}
  \E P(D,D)=\E|\caX(z,z')-\varkappa(z,z')|^{2D}\prec 1\,,
 \end{align}
for any $D$, hence by Markov's inequality we have
\begin{align}\label{le markov1}
 |\caX(z,z')-\varkappa(z,z')|\prec 1\,,
 \end{align}
which was to be proven for fixed $z,z'\in\mathcal{E}$.

It remains to extend this bound to a uniform bound for all $z,z'\in\mathcal{E}$. Let $\mathcal{L}\subset \mathcal{E}\times\mathcal{E}$ be a lattice such that $|\mathcal{L}|=O(N^{10})$ and for any $(z,z')\in\mathcal{E}\times\mathcal{E}$ there is a $(z_0,z_0')\in \mathcal{L}$ such that $|(z,z')-(z_0,z_0')|=O(N^{-10})$. Since $\langle\caH G(z)\caH G(z')\rangle$ is Lipschitz continuous in $(z,z')$ with constant bounded by $\eta_0^{-4}\le N^{4}$, $\eta_0=\min\{|\Im z|,\,|\Im z'|\}$, as follows from~\eqref{le basic rule derivative}, the uniform estimate follows from a union bound over $\mathcal{L}$ and~\eqref{le markov1}. This concludes the proof of Proposition~\ref{le proposition equation}, modulo the assumption that $\E h_{\iota,ij}^2=0$. This condition can easily be removed as we will show in Section~\ref{remark about variance complex case}.
\end{proof}

\begin{rem}\label{le remark about variance symmeric case}
 We can strengthen the estimate~\eqref{le proposition equation} to
 \begin{align*}
 \caX(z,z')&=m(z)m(z')+O_{\prec}(\Psi^2(z,z'))+O_{\prec}\big(\frac{1}{\sqrt{N}}\big)\\
 &=m_{sc}(z)m_{sc}(z')+O_{\prec}(\Psi^2(z,z'))+O_{\prec}\big(\frac{1}{\sqrt{N}}\big)\,.
\end{align*}
To establish this, one needs to count the number of off-diagonal Green function entries generated along the expansion procedure and then use $|G_{ij}|\prec \Psi+\delta_{ij}$.

 \end{rem}

\section{Real symmetric case}\label{le real symmetric case}
 
 In this section, we outline how our results for the complex Hermitian setup carry over to the real symmetric one. We start with the analogue to Assumption~\ref{assumption 1}.

\begin{assu}\label{assumption 2}
Fix an integer $k\ge 2$. We assume that $H_\iota:=(h_{\iota,ij})$ are $k$ independent real symmetric Wigner matrices of size $N\times N$, i.e., we assume that their entries are independent centred random variables, up to the symmetry constraints $h_{\iota,ij}={h_{\iota,ji}}$, satisfying
\begin{align}\label{le variance in real symmetric case}
\mathbb{E}h_{\iota,ij}^2=\frac{1+\delta_{ij}}{N}\,,\qquad 1\leq i,j\leq N\,,\qquad \iota=1,\ldots ,k\,,
\end{align}
and the families of random variables $\{h_{\iota,ij}\}$ have finite moments to all order, i.e., they satisfy~\eqref{entries moment bounds}.
\end{assu}
 We then have the following result for the real symmetric case.

 \begin{thm}\label{le theorem for symmetric}
  Let $H$ be given by~\eqref{le H} and assume that
 $H_\iota$, $\iota=1,\ldots, k$, satisfy Assumption~\ref{assumption 2} and that $\sigma_\iota$, $\iota=1,\ldots k$, satisfy~\eqref{le sigma condition}. Then
\begin{align}\label{le main theorem equation 2}
 \Big|w_\alpha^*H_\iota w_\beta-  \sigma_\iota\lambda_\alpha\delta_{\alpha\beta}\Big|\prec\frac{1}{\sqrt{N}}\,,
\end{align}
for all $\alpha,\beta\in\llbracket 1,N\rrbracket$ and $\iota\in\llbracket 1,k\rrbracket$.
\end{thm}
\nc

\begin{proof}
In the following we sketch the proof of Theorem~\ref{le theorem for symmetric} for $k=2$\nc. First, we define the cumulants, $\kappa_{\iota,ij}^{(p)}=\kappa_{\iota,ji}^{(p)}$ for the real random variables $h_{\iota,ij}$ as
\begin{align}
\kappa_{\iota,ij}^{(p)}:=(-\ii)^p\frac{\partial^p}{\partial s^p}\log \E\, \mathrm{e}^{\ii s h_{\iota,ij}}\Bigg|_{s=0}\,,
\end{align}
and note that they satisfy the estimate~\eqref{le size of kappa}. 

Second, we introduce the real symmetric analogue to $\caD_{ji}^{(p,q)}$ by setting
\begin{align}\label{le symmetric D}
 \caD_{ji}^{(p)}:= \frac{1}{p!}N^{\frac{p+1}{2}}\Big(\sigma_{2}\kappa_{1,ji}^{(p+1)}\partial_{1,ji}^p-\sigma_{1}\kappa_{2,ji}^{(p+1)}\partial_{2,ji}^p\Big)\,,\qquad\qquad p\in\N\,.
\end{align}

With these definitions we obtain the following cumulant expansion formula for the real symmetric case: Let $F$ be a monomial in the Green function entries and entries of $\caH$ as in Corollary~\ref{cor cumulant expansion}, then we have for any $l\in\N$,
\begin{align}\label{le cumulant expansion real case}
  \E \caH_{ij} F=\sum_{p=1}^l\frac{1}{N^{\frac{p+1}{2}}} \E \mathcal{D}_{ji}^{(p)}F+\E\Omega_{l+1}(F)\,,
 \end{align}
where the error term satisfies the bound
\begin{align}
|\E \Omega_{l+1}(F)|\prec N^{-(l+2)/2}\,.
\end{align}

Third, we recall that the basic differentiation rule for the real symmetric setup;
\begin{align}\label{le basic differential rule real}
  \partial_{\iota,ji}G(z)=-G(z)\sigma_{\iota}\Delta^{ji}G(z)-G(z)\sigma_{\iota}\Delta^{ij}G(z)\,,\qquad\quad \iota=1,2\,.
\end{align}
It is then easy to check that we have the computational rules
\begin{align}\label{le rule 2 real}
 \caD_{ji}^{(1)}G(z)=0\,,
\end{align}
as well as
\begin{align}\label{le rule 1 real}
 \caD_{ji}^{(1)}\caH_{ab}=\delta_{ja}\delta_{ib}+\delta_{ia}\delta_{jb}\,,
\end{align}
where $\caH_{ab}=\sigma_2H_1-\sigma_1H_2$ and where we used~\eqref{le sigma condition} and~\eqref{le variance in real symmetric case}.
 
 Armed with these definitions and rules, we turn to the computation of $\E\langle \caH G(z)\caH G(z')\rangle$. We follow the computation in Section~\ref{section expectation} up to~\eqref{le needed in real case too} that now becomes 
 \begin{align}\label{le where do we go}
\frac{1}{N^2} \sum_{ijab}\E\mathcal{D}_{ji}^{(1)}\Big[\Tr(\Delta^{ij}G\Delta^{ab}G')\caH_{ab}\Big]&=\frac{1}{N^2}\sum_{ijab}\E\Tr (\Delta^{ij}G\Delta^{ab}G')\mathcal{D}_{ji}^{(1)}\caH_{ab}\nonumber\\
& =\frac{1}{N^2}\sum_{ijab}\E\Tr (\Delta^{ij}G\Delta^{ab}G')(\delta_{ja}\delta_{ib}+\delta_{ia}\delta_{jb})\nonumber\\
 &=\E m(z)m(z')+\frac{1}{N^2}\sum_{ij}\E G_{ij}G_{ji}'\nonumber\\
 &=\E m(z)m(z')+O_\prec\big(\Psi(z,z')^2\big)\,,
 \end{align}
where we used the local law for the Green function in~\eqref{le local law} to get the last line; where we  used the fact that Theorem~\ref{theorem local law1} holds for real symmetric Wigner matrices as well. The only change was the additional error term $O_\prec\big(\Psi(z,z')^2\big)$ in~\eqref{le where do we go}. Following the computation in Section~\ref{section expectation} further, we conclude that Lemma~\ref{le lemma expectation final} holds in the real symmetric setup, too.

We move on to bound the higher moments of $\langle\caH G\caH G'\rangle$ following the arguments in Section~\ref{section proof of proposition}. Due to the modified rule~\eqref{le rule 1 real} in the real setup, we redefine $\varkappa(z,z')$ from~\eqref{le kappa rule} as
\begin{align}\label{le kappa real case}
 \varkappa(z,z'):=m(z)m(z')+\frac{1}{N^2}\sum_{ij}G_{ij}(z)G_{ji}(z')\,,
\end{align}
so that
\begin{align}\label{le kappa rule real case} 
\frac{1}{N^2} \sum_{ijab}\mathcal{D}_{ji}^{(1)}X^{ijab}\caH_{ab}=\varkappa(z,z')\,,
\end{align}
 holds with the adapted notation where $X^{ijab}$ given in~\eqref{le funny X}. This modification of $\varkappa$ ensures that $\caX(z,z')-\varkappa(z,z')$ is a self-normalizing quantity, i.e., in the computation of $\E P(D,D)$, with $P$ from~\eqref{le P recursive}, when some $\caH_{i_nj_n}$ acts on its twin $\caH_{a_nb_n}$ we get a zero contribution to $\E P(D,D)$ as in the complex Hermitian computation.
 
Yet, if some  $\caH_{i_nj_n}$ acts on another $\caH$ which is not its own twin, then we get an additional contribution from the second term on the right side of~\eqref{le rule 1 real} which is absent in the complex case. However, when this happens the number of free summation indices is reduced by two and we continue to expand the resulting term in the same way as in the complex case. Thus the modified rule~\eqref{le rule 1 real} produces more terms in the expansion of $\E P(D,D)$, but after all terms are fully expanded in the $\caH$'s, the size of the terms are estimated by the same power counting as in the complex Hermitian case. In this way one obtains that
\begin{align}
 \langle \caH \im G(z)\caH \im G(z')\rangle= \im m(z)\im m(z')+O_{\prec}(1)\,,
\end{align}
uniformly in~$z,z'\in\mathcal{E}$, similar to Proposition~\ref{le proposition}. The proof of Theorem~\ref{le theorem for symmetric} is then concluded in the same way as in Section~\ref{section proof of maintheorem}. \end{proof}

\section{Complex case revisited}\label{remark about variance complex case}

In this last section, we return to the complex Hermitian case. In the proof of Proposition~\ref{le main proposition} in Section~\ref{section proof of proposition}, we assumed for simplicity that $\E h_{\iota,ij}^2={\kappa_{\iota,ij}^{(0,2)}}=0$. In this section, we explain how this assumption can be removed. Even if ${\kappa_{\iota,ij}^{(0,2)}}\not=0$ and hence $\mathcal{D}_{ji}^{(0,1)}\not\equiv0$, we have
\begin{align}
 \caD_{ji}^{(1,0)}G(z)=\caD_{ji}^{(0,1)}G(z)=0\,,
\end{align}
similar to~\eqref{rule 2bis}. Further,~\eqref{rule 1} is modified as
\begin{align}
\caD_{ji}^{(0,1)}\caH_{ab}=\delta_{ib}\delta_{ja}\,,\qquad \caD_{ji}^{(1,0)}\caH_{ab}=\sigma_2^2\kappa_{1,ji}^{(0,2)}\delta_{ia}\delta_{jb}+\sigma_1^2\kappa_{2,ji}^{(0,2)}\delta_{ia}\delta_{bj}\,.
\end{align}
Since the cumulant expansions of Corollary~\ref{complex_cumulant} remain valid, it is straightforward to check that Proposition~\ref{le proposition} holds true also when ${\kappa_{\iota,ij}^{(0,2)}}$ do not necessarily vanish, after modifying the definition of $\varkappa$ similarly to the real symmetric case in order to obtain self-normalizing quantities in the moment bounds of $\caX-\varkappa$. More precisely, redefining
\begin{align}\label{le kappa complex case screw}
 \varkappa(z,z'):=m(z)m(z')+\frac{1}{N}\sum_{ij}\big(\sigma_2^2\kappa_{1,ji}^{(0,2)}G_{ij}(z)G_{ji}(z')+\sigma_1^2\kappa_{2,ji}^{(0,2)}G_{ij}(z)G_{ji}(z')\big)\,,
\end{align}
we find that
\begin{align}\label{le kappa rule complex case screw} 
\frac{1}{N^2} \sum_{ijab}\big(\mathcal{D}_{ji}^{(0,1)}+\mathcal{D}_{ji}^{(1,0)}\big)X^{ijab}\caH_{ab}=\varkappa(z,z')\,,
\end{align}
with $X^{ijab}$ given in~\eqref{le funny X}. We leave the further details aside. Finally, the proof of Theorem~\ref{le main theorem} from Proposition~\ref{le proposition} remains unaffected by this modification.

 \appendix
 
 \section{Proof of Lemma~\ref{lemma cumulant expansion}}\label{le appendix}
 
 \begin{proof}[Proof of Lemma~\ref{lemma cumulant expansion}] Fix the indices $i$ and $j$. We write $F\equiv F(h_{1,ij},h_{1,ji},h_{2,ij},h_{2,ji})$ to emphasize the explicit dependences. 
 From Lemma~\ref{complex_cumulant} and the definition of $\mathcal{D}^{(p,q)}_{ji}$ in~\eqref{le fancy D} we directly obtain~\eqref{le second cumulant expansion1} where $\Omega_{l+1}(F)$ is the sum of two error terms $\Omega_{1,l+1}$ and $\Omega_{1,l+2}$, the first coming from cumulant expansion with respect to $h_{1,ij}$, the second from expanding with respect to $h_{2,ij}$ in $\caH_{ij}$.  To bound the error term $\Omega_{1,l+1}$, we choose $M=N^{-1/4}$ in~\eqref{le error in complex cumulant}. Then together with the moment bounds in~\eqref{entries moment bounds}, for any (large) $D>0$, we have
 \begin{align}\label{le omega1l}
  |\Omega_{1,l+1}|&\le \frac{C_l}{N^{\frac{l+2}{2}}}\max_{p+q=l+1}\sup_{w\in\C, |w|\le N^{-1/4}}\big|\partial_w^p\partial_{\overline{w}}^q F(w,\overline{w},h_{2,ij},h_{2,ji})\big|\nonumber\\
  &\qquad+\frac{C_{l,D}}{N^{D/2}}  \Big(\E_{ij}\Big[\max_{p+q=l+1}\sup_{|w|\le |h_{1,ij}|} |\partial_w^p\partial_{\overline{w}}^qF(w,\overline{w}, h_{2,ij},h_{2,ji})|^2\Big]\Big)^{1/2}\,,
 \end{align}
 for $N$ sufficiently large, where we  used H\"older's inequality and the moment assumption~\eqref{entries moment bounds} to conclude that, for any $l$ and $D$,
\begin{align}
\E_{ij} |h_{1,ij}|^{l+2}\le \frac{C_{l}}{N^{\frac{l+2}{2}}}\,,\qquad\E_{ij} \Big[ |h_{1,ij}|^{2l+4} 1_{|h_{1,ij}|>N^{-1/4}}\Big]\le\frac{C_{l,D}}{N^{D/2}}\,,
\end{align}
for $N$ sufficiently large.

Consider next the Green function entry $G_{ab}\equiv G_{ab}(z)$ for some fixed $z\in\mathcal{E}$ and some choice of indices $a,b$. We write $G_{ab}=G_{ab}(h_{1,ij},h_{1,ji})$ in the following. Recall from the local law in~\eqref{le local law} that $|G_{ab}(h_{1,ij},h_{1,ji})|\prec \delta_{ab}+\Psi(z)$. Hence using a Neumann expansion of the resolvent we get
 \begin{align}\label{le little}
  G_{ab}(w,\overline{w})=G_{ab}(h_{1,ij},h_{1,ji})-\Tr \Big[\Delta^{ba}G(h_{1,ij},h_{1,ji})\sigma_{1}\big((w-h_{1,ij})\Delta^{ij}+(\overline{w}-h_{1,ji})\Delta^{ji}\big) G(w,\overline{w})\Big]\,.
 \end{align}
Thus denoting $$\hat\Lambda_o:=\sup_{w\in\C,|w|\le N^{-1/4}}\max_{a,b}|G_{ab}(w,\overline{w})|\,,$$ we get from~\eqref{le little} and $|h_{1,ij}|\prec\frac{1}{\sqrt{N}}$ that
\begin{align}
\hat\Lambda_o\prec 1+\frac{1}{N^{1/4}}\hat\Lambda_o\,,
\end{align}
hence $ \hat\Lambda_o\prec 1$. Next observe that $\partial_{ij}^p\partial_{ji}^q F$ is a polynomial in the Green function entries and the matrix entries of $\caH$ of degree at most $d+d'+l+2$.
Since $F$ was a monomial, 
the number of monomial summands in  $\partial_{ij}^p\partial_{ji}^q F$ depends on $p, q$ and $d+d'+l+2$
but is independent of $N$.
 Using the bounds $|\mathcal{H}_{ab}|\prec 1$ and $\hat\Lambda_o\prec 1$, we conclude that $\sup_{w\in\C, |w|\le N^{-1/4}}\big|\partial_{ij}^p\partial_{ji}^q F(w,\overline{w})\big|\prec 1$, for all $p,q$ with $p+q\le l+1$. It follows that the first term on the right side of~\eqref{le omega1l} is of order $O_\prec(N^{-(l+2)/2})$.

 To control the second term on the right side of~\eqref{le omega1l} we use once more that $\partial_{ij}^p\partial_{ji}^q F$ is finite linear combination of monomials in the Green function entries and the matrix entries of $\caH$. The maximal number of Green function entries occurring is $d+l+2$, estimating each factor by $\|G\|_\infty\le\frac{1}{\eta_0}\le N$, we get a contribution of order $N^{d+l+2}$ from the Green function entries. From the factors of $\caH_{ab}$ we use that $|\caH_{ab}|\prec 1$ and the moment bounds in~\eqref{entries moment bounds} to conclude that 
 \begin{align}
  \Big(\E_{ij}\big[\max_{p+q=l+1}\sup_{|w|\le |h_{1,ij}|} |\partial_{ij}^p\partial_{ji}^qF(w,\overline{w}, h_{2,ij},h_{2,ji})|^2\big]\Big)^{1/2}\prec N^{d+l+2}\,.
  \end{align}
  Hence choosing $D$ sufficiently large the second term on the right side in~\eqref{le omega1l} is bounded by $O_\prec(N^{-(l+2)/2})$. In sum, we have that $|\Omega_{1,l+1}(F)|\prec N^{-(l+2)/2}$. In the same way one derive the corresponding bound on $\Omega_{2,l+1}(F)$. 
\end{proof}

 \begin{proof}[Proof of Corollary~\ref{cor cumulant expansion}] Corollary~\ref{cor cumulant expansion} follows from Lemma~\ref{lemma cumulant expansion}, together with an application of item $(3)$ of Lemma~\ref{dominant} using the the estimates $|G_{ij}(z)|\prec 1$, $\|G(z)\|\le |\eta|^{-1}\le N$, and the moment assumptions in~Assumption~\ref{assumption 1} combined with H\"older's inequality.
\end{proof}

\end{document}